\documentclass[12pt]{article}
\usepackage{amsmath,amssymb}
\usepackage{amsthm} 
\usepackage{latexsym}
\usepackage{mathtools}
\usepackage[symbol]{footmisc}
\newtheorem{thm}{Theorem}[section]
\newtheorem{prop}[thm]{Proposition}
\newtheorem{lem}[thm]{Lemma}
\newtheorem{cor}[thm]{Corollary}

\begin{document}
\title{\bf On the broadcast routing problem in computer networks}
\date{}
\maketitle
\begin{center}
\author{
{\bf Brahim Chaourar \footnote[1]{This research was supported by the deanship of Scientific Research, Imam Mohammad Ibn Saud Islamic University, Saudi Arabia, Grant No 351223.}}
\\{Department of Mathematics and Statistics,\\Imam Mohammad Ibn Saud Islamic University (IMSIU) \\P.O. Box 90950, Riyadh 11623,  Saudi Arabia}
\\{email: bchaourar@hotmail.com}}
\end{center}

\begin{abstract}
Given an undirected graph $G = (V, E)$, and a vertex $r\in V$, an $r$-acyclic orientation of $G$ is an orientation $OE$ of the edges of $G$ such that the digraph $OG = (V, OE)$ is acyclic and $r$ is the unique vertex with indegree equal to 0. For $w\in \mathbb{R}^E_+$, $k(G, w)$ is the value of the $w$-maximum packing of $r$-arborescences for all $r\in V$ and all $r$-acyclic orientations $OE$ of $G$. In this case, the Broadcast Routing (in Computers Networks) Problem (BRP) is to compute $k(G, w)$, by finding an optimal $r$ and an optimal $r$-acyclic orientation. BRP is a mathematical formulation of multipath broadcast routing in computer networks. In this paper, we provide a polynomial time algorithm to solve BRP in outerplanar graphs. Outerplanar graphs are encountered in many applications such as computational geometry, robotics, etc.
\end{abstract}

\maketitle


\section{Introduction}

\textbf{Sets and their characteristic vectors will not be distinguished.} We refer to Bondy and Murty \cite{Bondy and Murty 2008} about graph theory terminology and facts.
\newline In computer networking, broadcasting refers to sending a packet to every destination simultaneously \cite{Tanenbaum 2002}.
Broadcast is a communication function that a node, called the source, sends messages to all the other nodes in the network. It is an important and complex function for implementation in optimal network routing with Quality of Service (QoS) requirement \cite{Chao et al. 2001, DeClerc and Paridaens 2002}.
\newline On the inverse of unicast routing, that is sending packets to one single destination, where several efficient multipath algorithms have been proposed \cite{Cetinkaya and Knightly 2004, Banner and Orda 2007, Gallager 1977, Ishida et al. 2006, Mahlous et al. 2008, Mahlous et al. 2009, Mahlous and Chaourar 2011}, few multipath algorithms for broadcast routing are known and most of the algorithms use a unique path to reach the destination, a spanning tree (or a rooted arborescence) \cite{Abolhasan et al. 2004, Akkaya and Younis 2005, Boukerche et al. 2011, Reina et al. 2015}. In this paper, we propose a novel and efficient multipath algorithm for broadcast routing in outerplanar networks.
\newline Multipath routing is proposed as an alternative to single path routing to take advantage of network redundancy, distribute load \cite{Suzuki and Tobagi 1992}, improve packet delivery reliability \cite{Ishida et al. 1992}, ease congestion on a network \cite{Bahk and Zarki 1992, Georgatsos and Griffin 1996}, improve robustness \cite{Tang and McKinley 2004}, increase network security \cite{Bohacek et al. 2002} and address QoS issues \cite{Bohacek et al. 2007}.
\newline Instead of finding the best spanning tree (or rooted arborescence) when using a single path, we propose to find an acyclic packing of rooted arborescences, i.e., directed rooted spanning trees, when using multipath routing because we should avoid corruption and redundancy, i.e., sending packets in two inverse directions at the same time and sending packets more than once to the same node, respectively, and we want to maximize the used bandwidth.
\newline An outerplanar graph is a graph that has a planar drawing for which all vertices belong to the outer face of the drawing. The outerplanar graphs are a subclass of the planar graphs. It follows that any simple and 2-connected outerplanar graph $G$ is formed by a circuit of the outer face and single edges linking between vertices of this circuit in a manner that preserves planarity. We call this circuit an {\em outer circuit} of $G$, these edges the {\em chords} of $G$ regarding this outer circuit, and the common vertices between the outer circuit and the chords the {\em join vertices}. We call the paths of the outer circuit joining successive join vertices the {\em outer chords}. Finally, we call the outer chords containing at least one vertex of degree 2, the {\em series outer chords}. The outer chords which are not series ones are single edges. Two chords (outer or not) are {\em adjacent} if they have one common join vertex.
\\The main result of this paper is a polynomial time algorithm to solve two mathematical formulations of multipath broadcast routing, namely BRP and rBRP, in outerplanar graphs (for the definitions of BRP and rBRP, see section 2).
\newline Many Combinatorial Optimization problems have been studied in outerplanar graphs/networks \cite{Andersen et al. 1998,  Calamoneri et al. 2009, Chang and Zhu 2009, Chen and Wang 2007, de Mier and Noy 2009, Esperet and Ochem 2007, Fiorini 1975, Fomin et al. 2005, Garg and Rusu 2007, Govindan et al. 1998, Gupta and Sinha 1999, Kant 1996, Kovacs and Lingas 2000, Liu and Lu 2010, Pinlou and Sopena 2006, Wang et al. 2002, Zhu and Goddard 1991} and most of them have been proved polynomial in this class of graphs. Another motivation is that "outerplanar graphs constitute an important class of graphs, often encountered in various applications, e.g., computational geometry, robotics, etc." \cite{Gupta and Sinha 1999}
\newline On the other hand, studying the broadcast routing problem in these outerplanar networks is also motivated by the difficulty of sending packets through "parallel" rooted arborescences in general graphs/networks without avoiding corruption or redundancy.
\newline In this paper, we use a combinatorial optimization approach and not a computer network one. However, we are motivated by the latter.
\newline The remainder of the paper is organized as follows: in section 2, we give a mathematical formulation of the problem, then, in section 3, we reduce the problem to simple and 2-connected graphs. We solve it in simple and 2-connected outerplanar graphs in section 4. And we conclude in section 5.

\section{Mathematical Formulation}
We say that a digraph $G=(V, E)$ is rooted at a vertex $r\in V$ if $r$ is the unique vertex in $G$ with indegree zero, that is there is no directed edge to $r$ (or $\delta^-(\{ r\})=\O$). In this case, the indegree of any vertex $v\in V$ is given by $|\delta^-(v)|$. An $r$-arborescence in the digraph $G$ is a spanning arborescence where $r$ is the root, i.e., the unique vertex with indegree zero. Given an undirected graph $G = (V, E)$, and a vertex $r\in V$, an acyclic orientation of $G$ rooted at $r$ (or an $r$-acyclic orientation) is an orientation $OE$ of the edges of G such that the digraph $OG = (V, OE)$ is acyclic and rooted at $r$. Moreover, if we define a weight function $w\in \mathbb{R}_+^E$, a $w$-packing of spanning trees (respectively, $r$-arborescences) is a finite class of pairs $(\lambda_i, T_i)$, $i=1, ..., t$, where the $\lambda_i$'s are positive reals and the $T_i$'s are spanning trees (respectively, $r$-arborescences), such that $\sum_{i=1}^{t} \lambda_i T_i (e)\leq w(e)$ for any edge $e$ ($T_i$ here is identified with its characteristic vector according to the context). The value of this $w$-packing is given by $\sum_{i=1}^{t} \lambda_i$. A $w$-maximum packing of spanning trees (respectively, $r$-arborescences) is a $w$-packing of spanning trees (respectively, $r$-arborescences) with a maximum value among all $w$-packings of spanning trees (respectively, $r$-arborescences).
\newline Given a graph $G=(V, E)$, and a nonnegative real weight function $w$ defined on $E$, we denote by $ku(G, w)$ the value of the $w$-maximum packing of spanning trees in $G$.
\\Given a vertex $r\in V$, and an $r$-acyclic orientation $OE$ of the edges of $G$, we denote by $k(G, OE, w, r)$ the value of the $w$-maximum packing of $r$-arborescences in $OG = (V, OE)$ for a fixed $r\in V$, $k(G, w, r)=Max \{k(G, OE, w, r) : OE$ an $r$-acyclic orientation of  $G$\}, that is the value of $w$-maximum packing of $r$-arborescences among all possible $r$-acyclic orientations of $G$ for a fixed $r\in V$, and $k(G, w) = Max\{k(G, w, r): r\in V\}$, that is the value of $w$-maximum packing of $r$-arborescences among all $r\in V$ and all possible $r$-acyclic orientations of $G$. We say that a vertex $r$ and an $r$-acyclic orientation $OE$ are optimal if $k(G, OE, w, r)=k(G, w)$.
\newline Now we can define the Broadcast Routing Problem (BRP) as follows. Given an undirected graph $G = (V, E)$, and a nonnegative real weight function $w$ defined on $E$. The question is to compute $k(G, w)$ by finding appropriate optimal root $r$ and $r$-acyclic orientation $OE$, and the corresponding $r$-arborescences involved in the maximum packing.
\newline $OG = (V, OE)$ should be acyclic to avoid corruption and redundancy, and the root $r$ represents the source node in the broadcast routing.
\newline We route data through $OG$ by using the optimal packing of $r$-arborescences, thus we maximize the used bandwidth for each directed edge of $OG$ and, virtually, we have a multipath broadcast routing which is given by the involved $r$-arborescences in the maximum packing.
\newline We also need to define a related problem rBRP as follows. Given an undirected graph $G = (V, E)$, a vertex $r\in V$, and a nonnegative real weight function $w$ defined on $E$. The question is to compute $k(G, w, r)$ by finding an appropriate optimal $r$-acyclic orientation $OE$, and the corresponding $r$-arborescences involved in the maximum packing. It is clear that BRP can be reduced to rBRP in $O(n)$, where $n=|V|$. rBRP modelizes a classical broadcasting problem where the source $r$ is given in the input, while BRP modelizes a new point of view where we want to find the best source for broadcasting in a given network. We solve both problems in outerplanar graphs.
\newline Given a directed graph $G=(V, E)$, an inedge to a given vertex $v\in V$, is any directed edge $e\in \delta^-(v)$.
\newline Since we are dealing with maximum packing of rooted arborescences, we need the following theorem due to Edmonds \cite{Edmonds 1970}. Note that, given a digraph $G=(V, E)$ and $r\in V$, an $r$-cut is a cut $\delta^+(U)$ (respectively, $\delta^-(U)$ with $U\neq \O$) for some $U\subset V$ with $r\in U$ (respectively, $r\notin U$).
\begin{thm}
Given a digraph $G$, a vertex $r\in V(G)$, and a weight function $w\in \mathbb{R}_+^{E(G)}$, then the value of the $w$-maximum weight packing of $r$-arborescences equals the value of the $w$-minimum weight of an $r$-cut.
\end{thm}
Note that $ku(G, w)$ and $k(G, w)$ are not equal in general, even for (circuits and) outerplanar graphs as shown in the minimal example here below.
\newline \textbf{Example:} Let $G=abcda$ be a circuit on 4 vertices such that $V(G)=\{a, b, c, d\}$, $w(\{a, b\}) = w(\{c, d\}) = 2$, and $w(\{a, d\}) = w(\{b, c\}) = 1$.
\newline It is not difficult to see that $ku(G, w) = 2$ and $k(G, w) = 1$ because for any chosen root $r$ and any $r$-acyclic orientation of $G$, there is a trivial $r$-cut with weight 1, and, according to Theorem 2.1, the value of the maximum packing of $r$-arborescences equals the value of the minimum weight of an $r$-cut.
{\hfill \hbox{\raisebox{.5ex}{\fbox}$\phantom{.}$}}
\newline We give here a result about the hardness of getting a solution by enumeration.
\begin{prop}
The number of rooted acyclic orientations of $K_n$ ($n\geq 2$) is $n!$.
\end{prop}
\begin{proof}
Since every acyclic orientation gives a topological order of the vertices and vice-versa, and there is a unique acyclic orientation if we fix an order for the vertices, then the number of rooted acyclic orientations of $K_n$ is equal to the number of orders on its vertices which is $n!$.
\end{proof}


\section{BRP is reducible to simple and 2-connected graphs}

First we can reduce BRP and rBRP to simple graphs by using the following proposition.
\begin{prop} Consider the two following cases:
\newline (1) Let $G = (V, E)$ be an undirected graph, $e$ and $f$ are parallel edges of $G$, and $w\in \mathbb{R}_+^E$. Let $G_f = G\backslash f$ and $w_f\in \mathbb{R}_+^{E\backslash \{f\}}$ such that $w_f(e) = w(e)+w(f)$ and $w_f(g) = w(g)$ for $g\neq e$.
\newline (2) Let $G = (V, E)$ be an undirected graph, $f$ is a loop of $G$, and $w\in \mathbb{R}_+^E$. Let $G_f = G\backslash f$ and $w_f\in \mathbb{R}_+^{E\backslash \{f\}}$ such that $w_f(g) = w(g)$ for $g\neq f$.
\newline Then for both cases (1) and (2), we have: $k(G, w) = k(G_f, w_f)$ and $k(G, w, r) = k(G_f, w_f, r)$ for any $r\in V$.
\end{prop}
\begin{proof}
Direct from the fact that parallel edges in acyclic digraphs should be oriented in the same direction, and removing loops do not affect any packing of rooted arborescences because arborescences do not contain loops.
\end{proof}
We also can reduce BRP and rBRP to 2-connected graphs by using the following proposition. We denote by $G_1\oplus_v G_2$, or $G_1\oplus G_2$ if there is no confusion, the 1-sum of two graphs $G_1$ and $G_2$ where $v$ is the base vertex for this 1-sum.
\begin{prop}
Let $G = (V, E)$ be an undirected graph, such that $G=G_1\oplus G_2$, where $v\in V(G_1)\cap V(G_2)$ is the base vertex for the 1-sum, $v_1\in V(G_1)$, $w\in \mathbb{R}_+^E$, and $w_i\in \mathbb{R}_+^{E(G_i)}$, $i=1,2$, be the projection of $w$ on $E(G_i)$, $i=1,2$.
\newline Then $k(G, w, v_1) = Min\{ k(G_1, w_1, v_1), k(G_2, w_2, v)\}$.
\end{prop}
\begin{proof}
Direct from the fact that the restriction of any $v_1$-acyclic orientation of $G$ to $G_i$, $i=1,2$, induces a $v_1$-acyclic orientation of $G_1$ and a $v$-acyclic orientation of $G_2$, respectively, and any packing of $v_1$-arborescences in $G$ induces a packing of $v_1$-arborescences in $G_1$ and a packing of $v$-arborescences in $G_2$.
\end{proof}
We need the following lemma for solving BRP and rBRP in any kind of graphs.
\begin{lem}
Given an acyclic digraph $G$, a vertex $r\in V(G)$, and a weight function $w\in \mathbb{R}_+^{E(G)}$, then the value of the $w$-minimum $r$-cut is reached at a (single vertex) trivial $r$-cut.
\end{lem}
\begin{proof}
Let $U\subset V(G)$ such that $r\not\in U$ and $U$ is not empty. Then the subgraph $G[U]$ induced by $U$ is also acyclic and contains at least one (root) vertex $v\in U$ such that the indegree of $v$ in $G[U]$ is zero. It follows that $\delta^-(v)\subseteq \delta^-(U)$ and $w(\delta^-(v))\leq w(\delta^-(U))$ because the weights are nonegative. We can conclude by taking the minimum of both sides of the later inequality.
\end{proof}
In this case, we denote by $\delta_{G, r, w}$, or $\delta_{r, w}$ if  there is no ambiguity, the value of the $w$-minimum trivial $r$-cut in a digraph $G$ for a given weight function $w\in \mathbb{R}_+^{E(G)}$.
We can deduce here below how to compute a maximum packing of rooted arborescences in acyclic digraphs. Recall that an acyclic digraph is rooted at a vertex $r$ if the indegree of $r$ is zero.
\begin{cor}
Given an acyclic digraph $G$ rooted at the vertex $r\in V(G)$, and a weight function $w\in \mathbb{R}_+^{E(G)}$, then the $w$-maximum packing of $r$-arborescences can be computed in $O(m^2)$ where $m = |E(G)|$.
\end{cor}
\begin{proof}
We give here an $O(m^2)$ algorithm for this purpose.
\newline \textbf{Input:} An acyclic digraph $G = (V\cup \{r\}, E)$ rooted at $r$, and $w\in \mathbb{R}_+^{E(G)}$.
\newline \textbf{Output:} Pairs $(\lambda_1, A_1), ...,(\lambda_t, A_t)$ of positive reals and $r$-arborescences such that $\sum_{i=1}^{t} \lambda_i A_i (e)\leq w(e)$, for any edge $e\in E$ and $\sum_{i=1}^{t} \lambda_i$ is maximum.
\newline (1) Set $i := 1$.
\newline (2) Find $e_0\in E$ such that $w(e_0) = Min\{w(e) : e\in E\}$.
\newline (3) If $w(e_0) > 0$ then do:
\newline (3.1) set $\lambda_i := w(e_0)$,
\newline (3.2.1) if we can complete $\{e_0\}$ to an $r$-arborescence $A$ by picking one edge $e\in \delta^-(v)$ for all $v\in V$ such that $w(e) > 0$, then  $A_i := A$,
\newline (3.2.2) else: stop.
\newline (3.3) set $w(e) := w(e)-w(e_0)$ for all $e\in A_i$, and set $i := i+1$,
\newline (3.4) go to (2).
\newline (4) End of algorithm.
\newline Since, in each iteration, there is at least one edge $e_0$ whose weight becomes zero (step 3.3), then we need at most $m$ iterations to get the output. In each iteration, we need to get the minimum weight edge in $O(m)$ (step 2), to complete one edge to a rooted arborescence in $O(n)$ (step 3.2.1) , where $n = |V|$ (because the given graph is acyclic so we need only to choose an inedge for any vertex if it exists), and to update the weight function in $O(n)$ (step 3.3). So the whole running time is $O(m^2)$.
\newline Now, for the correctness of the algorithm, we will prove it by induction on the cardinality of support($w$) = $\{e\in E$ such that $w(e)>0\}$. Without loss of generality, we can suppose that $support(w)$ contains an $r$-arborescence because otherwise there is no packing. If $|support(w)|=n-1$ then $\lambda =min\{ w(e)$ such that $e\in E\}$ and $A=support(w)$ give the solution. Suppose now that $support(w)\geq n$.
\newline Let $w'\in \mathbb{R}_+^{E(G)}$ be the weight function obtained from $w$ after the first iteration of the algorithm, i.e., after step (3.3) for the first time. It is clear that $|support(w')|\leq |support(w)|-1$. It follows that $w'=\sum_{i=2}^{t} \lambda_i A_i$ is a $w'$-maximum packing of $r$-arborescences in $G$. According to Theorem 2.1 and Lemma 3.3, $\delta_{r, w'}=\sum_{i=2}^{t} \lambda_i$. Now, for any $v\in V$, $w(\delta^-(v))=w'(\delta^-(v))+w(e_0)$. By taking the minimum of both sides, we have $\delta_{r, w}= \delta_{r, w'}+w(e_0)=\sum_{i=2}^{t} \lambda_i + \lambda_1$, which means that our packing is maximum.
\end{proof}

\section{BRP is polynomial in (simple and 2-connected) outerplanar graphs}

According to the previous section, it suffices now to prove that BRP and rBRP are polynomial in simple and 2-connected outerplanar graphs. First, we solve BRP and rBRP in circuits, which are basic simple and 2-connected outerplanar graphs.
\begin{lem}
BRP and rBRP can be solved in $C_n$, the circuit on $n$ vertices, with $n\geq 2$, in linear time $O(n)$.
\end{lem}
\begin{proof}
Let $G = C_n$ be a circuit on $n$ vertices with $n\geq 2$, and a weight function $w\in \mathbb{R}_+^{E(G)}$. Since the case $n = 2$ is trivial then we can suppose that $n\geq 3$. First we find the three $w$-minimum edges $w(e_1)\leq w(e_2)\leq w(e_3)\leq w(e)$, for any $e\in E(C_n)\backslash \{e_1, e_2, e_3\}$.
\newline Given $r\in V(G)$, it is not difficult to see that, in any $r$-acyclic orientation of $G$, there is exactly one sink $s\in V(G)$, i.e., a vertex with outdegree equal to zero (or $\delta^+ (\{ s\})=\O$). We say that such a sink is optimal if it is induced by an optimal $r$-acyclic orientation. According to Theorem 2.1 and Lemma 3.3, the value of the $w$-maximum packing of $r$-arborescences equals the value of the $w$-minimum trivial $r$-cut. Thus, there is an edge $e$, where $w(e)=w(e_1)$ (if $w(e_1)<w(e_2)$ then necessary $e=e_1$), which is incident to any (optimal) sink $s$ for any optimal $r\in V$ and any optimal $r$-acyclic orientation of $G$, because if $OE$ (respectively, $O'E$) is an $r$-acyclic orientation of $G$ such that there is such an edge (respectively, all such edges) $e$ is (respectively, are) incident (respectively, not incident) to the unique sink of $OG=(V, OE)$ (respectively, of $O'G=(V, O'E)$), then $k(G, w)\geq k(G, OE, w, r)=\delta_{OG, r, w}\geq w(e_2)\geq w(e_1)=\delta_{O'G, r, w}=k(G, O'E, w, r)$. Without loss of generality, we can suppose that $e=e_1$. We have then the following cases:
\newline \textbf{Case 1:} If $e_2$ is adjacent to $e_1$ then the optimal sink $s$ should be incident to both edges, $k(G, w) = Min\{w(e_1)+w(e_2), w(e_3)\}$ and any vertex distinct from $s$ can be an optimal root.
\newline \textbf{Case 2:} If $e_2$ is not adjacent to $e_1$ then let $e$ and $f$ be the two edges adjacent to $e_1$ such that $w(e)\leq w(f)$. Thus, the optimal sink $s$ should be incident to $e_1$ and $e$, $k(G, w) = Min\{w(e_1)+w(e), w(e_2)\}$ and any vertex distinct from $s$ can be an optimal root.
\end{proof}
For general simple and 2-connected outerplanar graphs, we need the following results.
\begin{lem}
Let $G$ be a simple and 2-connected outerplanar graph with $k$ chords. If $S_2(G)$ is the maximum number of sinks with degree 2, for all $r\in V$ and all $r$-acyclic orientations of $G$, then $S_2(G)\leq k+1$.
\end{lem}
\begin{proof}
Let $r\in V$ and $OG=(OE, V)$ be a root and an $r$-acyclic orientation for which the number of sinks with degree 2 is maximum, i.e., equals $S_2(G)$. We will prove that $S_2(G)\leq k+1$ by induction on $k$.
\newline If $k=0$ then $G$ is a circuit and trivially any acyclic orientation contains exactly one sink.
\newline Suppose now that $k\geq 1$.
\newline Note that all chords of $G$ do not contain sinks of degree 2 in $OG$.
\newline Let $C$ be an outer circuit of $G$ and $s\in V(C)$ be a sink of degree 2 in $OG$. It follows that there is a directed $rs$-path in $OG$ passing through at least one chord $e$ (Otherwise if no chord is used in any $rs$-path, for all sinks $s\in V$ of degree 2, we can remove this chord and, by induction on $G\backslash e$, we get the inequality). We choose $e$ as the last chord in that $rs$-path. Let $v$ be the last join vertex of $e$ in that $rs$-path. Let now $O'G$ be the $r$-acyclic orientation of $G$ obtained from $OG$ by reversing the orientation of all edges of the unique $vs$-path and maintaining the orientation of all remaining edges.
\newline {\bf Case 1:} $v$ is a sink in $O'G$.
\newline It follows that $r$ and $O'G\backslash e$ are optimal for $G\backslash e$ regarding the number of sinks of degree 2, because, otherwise, we will get a better maximum for $G$, which contradicts our assumption that $r$ and $OG$ induces a maximum number of sinks of degree 2. Thus $S_2(G)=S_2(G\backslash e)\leq k \leq k+1$.
\newline {\bf Case 2:} $v$ is not a sink in $O'G$.
\newline It follows that the number of sinks of degree 2 in $O'G\backslash e$ is $S_2(G)-1$. Thus $S_2(G)-1\leq S_2(G\backslash e)\leq k$, and then, $S_2(G)\leq k+1$.
\end{proof}
\begin{lem}
Let $G$ be a simple and 2-connected directed acyclic outerplanar graph rooted at $r\in V(G)$ with $k$ chords, and $s_i(G)$, $i=2,3, ..., \Delta(G)$ ($\Delta (G)$ is the maximum degree of $G$), be the number of sinks in $G$ with degree $i$. Then: $$\sum_{i=2}^{\Delta (G)} (i-1)s_i(G)\leq k+1.$$
\end{lem}
\begin{proof}
By induction on $S(G)=\sum_{i=3}^{\Delta (G)} s_i(G)$.
\newline If $S(G)=0$ then, according to Lemma 4.2, $\sum_{i=2}^{\Delta (G)} (i-1)s_i=s_2(G)\leq k+1$.
\newline If $S(G)\geq 1$ then there exists $i_0\in \{3, ..., \Delta(G)\}$ such that $s_{i_0} (G)\geq 1$. It follows that there is a sink $v_0$ in $G$ with degree $i_0\geq 3$. Let $P$ be the union of all chords adjacent to $v_0$ in $G$. Note that all chords in $P$ are one-way directed to $v_0$. Note also that the number of chords in $P$ is $i_0-2$.
\newline Let $G'=G\backslash P$. Since $G'$ is still $r$-acyclic, then by induction on $S(G')=S(G)-1$, we have $\sum_{i=2}^{\Delta (G')} (i-1)s_i(G')\leq k-i_0+3$ because there are $k-(i_0-2)$ chords in $G'$. On the other hand, since all $s_i(G)=s_i(G')$ except $s_{i_0}(G)=s_{i_0}(G')+1$, and $s_2(G)=s_2(G')-1$ ($v_0$ becomes a sink of degree 2 in $G'$), we have $\sum_{i=2}^{\Delta (G)} (i-1)s_i(G)=\sum_{i=2}^{\Delta (G')} (i-1)s_i(G')+(i_0-1)-1\leq (k-i_0+3)+(i_0-2)=k+1$, and we are done.
\end{proof}
To solve BRP and rBRP in (simple and 2-connected) outerplanar graphs, we need to give some definitions. A mixed graph $G=(V, E_u, E_d)$ is a graph where the edges of $E_u$ are undirected and the (remaining) edges of $E_d$ are directed. In this case, we use the notations $\delta_G (\{ v \}) \subseteq E_u$ and $\delta_G^- (\{ v \}) \subseteq E_d$. A partial $r$-acyclic orientation of an undirected graph $G=(V, E)$, where $r\in V(G)$, is a mixed graph $G'=(V, E_u\cup E_d)$ such that: (1) The underlying graph of $G'$ is $G$; (2) There exists an orientation $OE_u$ of $E_u$ such that the obtained graph $G''=(V, OE_u\cup E_d)$ is an $r$-acyclic orientation of $G$ 
\newline Note that any undirected graph and any of its acyclic orientations are particular cases of partial acyclic orientations ($E_d=\O$ and $E_u=\O$, respectively).
\newline We call an undirected edge of a partial acyclic orientation a {\em forced edge} if it has one unique possible orientation in order to get (at the end) an acyclic orientation. For example, if we choose the root $r$, then all edges incident to $r$ are forced because they should be oriented from $r$ to its adjacent vertices in order to get at the end an $r$-acyclic orientation. Another example is when we choose a sink of degree 2 in an outer chord, then all edges of this outer chord are forced because they should be oriented in the direction of that sink (except if we have a root of degree 2 in the same outer chord). In opposition, a non forced edge is called a {\em free} edge. We have then the following condition for forced edges.
\begin{lem}
Let $G'=(V, E_u\cup E_d)$ be a partial $r$-acyclic orientation of an undirected simple and 2-connected outerplanar graph $G=(V, E)$ with $k$ chords and $r\in V$. If $$\sum_{v\in V} (|\delta_{G'}^-(v)|-1) = k+1,$$ then all edges of $E_u$ are forced.
\end{lem}
\begin{proof}
We consider the function $f(G')=\sum_{v\in V} (|\delta_{G'}^-(v)|-1)$. We will prove by induction on $k$, that if $f(G')=k+1$ then all edges of $E_u(G')$ are forced.
\newline If $k=0$ then $G$ is a circuit and the edges become forced if we choose the root and the sink, i.e., if $f(G')=1$.
\newline Suppose now that $k\geq 1$ and $f(G')=k+1$.
\newline {\bf Case 1:} There exists $v\in V$ such that $|\delta_{G'}^-(v)|\geq 2$ and $|\delta_G(v)|\geq 3$.
\newline Let $e$ be a chord for which $v$ is a join vertex.
\newline {\bf Subcase 1.1:} $e\in \delta_{G'}^-(v)$.
\newline It follows that $G'_e=G'\backslash e$ is a partial $r$-acyclic orientation of $G_e=G\backslash e$ with $f(G'_e)=k$. Thus all remaining edges of $G'_e$ are forced.
\newline {\bf Subcase 1.2:} $e\in \delta_{G'}^+(v)$.
\newline Let $u\neq v$ be the second join vertex of $e$ then $e\in \delta_{G'}^-(u)$. If $u$ contributes to $f(G')$ then $|\delta_{G'}^-(u)|\geq 2$ and $|\delta_G (u)|\geq 3$, and we are in Subcase 1.1. Otherwise, $G'$ is not acyclic, a contradiction.
\newline {\bf Case 2:} For all $v\in V$ such that $|\delta_{G'}^-(v)|\geq 2$, we have: $|\delta_G(v)|=2$.
\newline It follows that all vertices contributing to $f(G')$ are sinks of degree 2. Since the contribution of each such sink is 1 then their number is $S_2(G)=f(G')=k+1$, according to Lemma 4.2.
\newline Since all chords containing sinks of degree 2 in $G'$ are outer chords, then it is not difficult to see that the number of outer chords is $k+1$ because chords should be incident successively. So all undirected edges of the outer circuit are forced. Suppose now, by contradiction, that there exists a free undirected edge for some chord $e$. So $G''=G'\backslash e$ is a partial $r$-acyclic orientation of $G\backslash e$ with $k-1$ chords and $S_2(G\backslash e)=f(G'')=k+1$, a contradiction with Lemma 4.2.
\end{proof}
\begin{cor}
The subgraph formed by all forced chords is a branching, i.e., a directed subgraph where all vertices have an indegree at most 1 except the root.
\end{cor}
\begin{proof}
Direct from the fact that there is no contribution to the function $f$ (introduced in the last proof) induced by forced edges, so all vertices incident to inedges of these forced edges should have indegree 1 in the (full) acyclic orientation (at the end).
\end{proof}
Now, we are ready for the main result of this paper.
\begin{thm}
BRP and rBRP can be solved in simple and 2-connected outerplanar graphs in polynomial time.
\end{thm}
\begin{proof}
We give hereinbelow a polynomial time algorithm for solving rBRP in simple and 2-connected outerplanar graphs.
\newline \textbf{Input:} An undirected simple and 2-connected outerplanar graph $G = (V, E)$ with $k$ chords, $r\in V$, and $w\in \mathbb{R}_+^E$.
\newline \textbf{Output:} $G'=(V, E_d)$ an $r$-acyclic orientation of $G$ such that: $$k(G, w, r)=Min\{ \delta_{G'}^-(v) : v\in V\backslash \{ r\}\}.$$
 (0) Orient all current forced edges (incident to $r$) and set $F:=\delta (\{ r\})$.
\newline (1) Set $E_u := E\backslash F$, $E_d:=F$, $G'':=(V, E_u, E_d)$ (a mixed graph), $t:=1$, $U:=V$ and $f:=0$.
\newline (2) While $f < k+1$ do:
\newline (2.1) Find $u_t\in U\backslash \{ r\}$ such that $\delta_t:=w(\delta_{G''}(u_t)\cup \delta_{G''}^-(u_t)) = Min\{w(\delta_{G''}(v)\cup \delta_{G''}^-(v)) : v\in U\backslash \{ r\}\}$.
\newline (2.2) Set $f:=f+|\delta_{G''}(u_t)\cup \delta_{G''}^-(u_t)|-1$.
\newline (2.3) Orient all edges of $\delta_{G''}(u_t)$ in the direction of $u_t$, and all current trivial forced edges $F'$.
\newline (2.4) Set $F:=F\cup F'$, $E_u:=E_u\backslash (F\cup \delta_{G''}(u_t))$, $E_d:=E_d\cup F\cup \delta_{G''}(u_t))$, and $U:=U\backslash \{ u_t\}$.
\newline (2.5) Set $t:=t+1$.
\newline (2.6) End of While.
\newline (3) Set $\delta_t:=Min\{ w(e) : e\in E_u\}$.
\newline (4) Orient all edges of $E_u$ as a branching with a root $r$.
\newline (5) $G'=(V, E_d)$.
\newline (6) End of algorithm.
\newline Step (2.1) has a running time complexity $O(m)$, where $m=|E|$. Since we will run the loop of While at most $k$ times, then the running time of the whole algorithm is at most $O((k+1) m)$ because we need to compute a minimum once in step (3).
\newline Now, we will prove the correctness of the algorithm by induction on $k$.
\newline Note that Lemma 4.4 and Corollary 4.5 imply that the obtained directed graph $G'$ is an $r$-acyclic orientation of $G$.
\newline  Let $T$ be the number of iterations of the While loop. It is not difficult to see that: $$Min\{ \delta_{G'}^-(v) : v\in V\backslash \{ r\}\}=Min\{ \delta_t : 1\leq t\leq T+1\}.$$
 If $k=0$ then $G$ is a circuit and according to Lemma 4.1, $k(G, w, r)=Min\{ \delta_1, \delta_2\}$.
\newline Suppose now that $k\geq 1$ and $\delta_{t_0} = Min\{ \delta_t : 1\leq t\leq T+1\}$ where $t_0$ is the smallest possible one. We can suppose that $t_0\geq 2$ because if $t_0=1$ then we can conclude.
\newline It is not difficult to see that there is an outer chord $P$ such that $r$ is not an inner vertex of $P$ and $G\backslash P$ is a simple and 2-connected outer planar graph. Let $e=\{ u, v\}$ be the chord adjacent to $P$ in $G$ such that $e$ belongs to an outer chord in $G\backslash P$. Suppose that $e$ is oriented as $(u, v)$ in $G'$.
\newline {\bf Case 1:} There is a sink $v_t\in V(P)$ of degree 2 in $G'$.
\newline Note that all edges of $P$ are oriented in the direction of $v_t$. Since $t_0\geq 2$ then $t\neq t_0$ and $\delta_t > \delta_{t_0}$. It follows that $G'\backslash P$ is the output of the algorithm for the input $G\backslash P$ and $w'$, the weight function $w$ restricted to $E\backslash P$. By induction on $k$, we have: $\delta_{t_0} = k(G\backslash P, w', r)$.
\newline Whatever the orientation of the edges of $P$, $\delta_{t_0}$ cannot be improved, so $\delta_{t_0} = k(G, w, r)$.
\newline {\bf Case 2:} There is no sink $v_t\in V(P)$ of degree 2 in $G'$.
\newline Note that all edges of $P$ are oriented in the direction of $v$ in $G'$. Let $f\in P$ be the incident edge to $v$, $w_P=Min\{ w(g) : g\in P\}$, and $w'\in \mathbb{R}_+^{E\backslash P}$ such that $w'(e)=w_P+w(e)$ and $w'(g)=w(g)$ otherwise. It is not difficult to see that $G'\backslash P$ is the output of our algorithm for the input $G\backslash P$ and $w'$. So by induction on $k$, we have $k(G, w, r)=k(G\backslash P, w', r)=\delta_{t_0}$.
\end{proof}

\section{Conclusion}
We have introduced two mathematical formulations of multipath broadcast routing (in computer networks), namely BRP and rBRP. Then we have solved BRP and rBRP in outerplanar graphs. Future investigations can be studying the running time complexity of BRP and rBRP in general graphs and solve them in larger classes than outerplanar graphs.


\end{document}